\documentclass[11pt]{article}

\usepackage[utf8]{inputenc}
\usepackage{amsmath}
\usepackage[english]{babel}
\usepackage{amssymb}
\usepackage{amsthm}
\usepackage{mathtools}
\usepackage{esdiff}
\usepackage{microtype}
\usepackage{paralist}
\usepackage{skak}
\usepackage{bbm}
\usepackage[paper=a4paper,left=25mm,right=25mm,top=25mm,bottom=25mm]{geometry}
\usepackage{stmaryrd}
\usepackage{hyperref}
\usepackage{color} 
\usepackage{fancyhdr}
\usepackage{nomencl}
\usepackage{cite}
\usepackage{longtable}
\usepackage{thmtools}
\usepackage{cuted}
\usepackage{MnSymbol}
\usepackage{tikz-cd}

\declaretheoremstyle[headfont=\normalfont\bfseries]{bfthmstyle}

\declaretheorem[sharenumber=Theorem,style=bfthmstyle]{Lemma}
\declaretheorem[sharenumber=Theorem,style=bfthmstyle]{Proposition}

\declaretheoremstyle[headfont=\normalfont\bfseries,qed={$\diamond$}]{bfdefstyle}
\declaretheorem[sharenumber=Theorem,style=bfdefstyle]{Definition}
\declaretheorem[sharenumber=Theorem,style=bfdefstyle]{Definition-Proposition}
\declaretheorem[sharenumber=Theorem,style=bfdefstyle]{Definition-Lemma}
\declaretheorem[sharenumber=Theorem,style=bfdefstyle]{Example}

\declaretheoremstyle[headfont=\normalfont\bfseries,qed={$\diamond$}]{bfremstyle}
\declaretheorem[sharenumber=Theorem,style=bfdefstyle]{Remark}

\newcommand{\set}[1]{\left\lbrace #1 \right\rbrace}

\declaretheoremstyle[
notefont=\normalfont, notebraces={}{},
headformat=\mathbb{N}UMBER~\mathbb{N}AME~\mathbb{N}OTE
]{nopar}

\numberwithin{equation}{section}

\makenomenclature

\author{Jonas Kirchhoff, Bernhard Maschke}
\title{Remarks on the geometric structure of port-Hamiltonian systems}
\date{\today}

\begin{document}
\setlength{\parindent}{0em}
\pagestyle{fancy}
\lhead{Jonas Kirchhoff, Bernhard Maschke}
\rhead{Geometric structure}

\maketitle

\paragraph{Abstract} We study the geometric structure of port-Hamiltonian systems. Starting with the intuitive understanding that port-Hamiltonian systems are ``in between'' certain closed Hamiltonian systems, the geometric structure of port-Hamiltonian systems must be ``in between'' the geometric structures of the latter systems. These are Courant algebroids; and hence the geometric structures should be related by Courant algebroid morphisms. Using this idea, we propose a definition of an intrinsic geometric structure and show that it is unique, if it exists.\vspace{-1mm}
 
\vspace{-1mm}

\vfill
\par\noindent\rule{5cm}{0.4pt}\\
\begin{footnotesize}
Corresponding author: Jonas Kirchhoff\\[1em]
Jonas Kirchhoff\\
Institut für Mathematik, Technische Universität Ilmenau, Weimarer Stra\ss e 25, 98693 Ilmenau, Germany\\
E-mail: jonas.kirchhoff@tu-ilmenau.de\\[1em]

Bernhard Maschke\\
Univ. Lyon, Université Claude Bernard Lyon 1, CNRS, LAGEPP UMR 5007, France\\
E-mail: bernhard.maschke@univ-lyon1.fr\\[1em]

Jonas Kirchhoff thanks the Technische Universität Ilmenau and the Freistaat Thüringen for their financial support as part of the Thüringer Graduiertenförderung, and the French-Dutch-German Doctoral College Port-Hamiltonian Systems: Modeling, Numerics and Control. Bernhard Maschka acknowledges support by IMPACTS (ANR-21-CE48-0018).

\end{footnotesize}

\newpage

\section{Introduction}

Port-Hamiltonian systems have become rather popular since their introduction in~\cite{MascScha92}. They have since sparked a rich theory for open systems whose closed counterparts are covered by the classical Hamiltonian mechanics. These formulations utilise a \textit{Dirac structure} on the underlying extended state space $\mathbb{R}^n\oplus (\mathbb{R}^n)^*$ in the finite-dimensional case (to which we restrict ourself for simplcity of exposition) as the structure which describes the power transport in the system; in particular, the (network) topology of the system is contained in the Dirac structure. In the case that the underlying topology is not constant (e.g. we may have a state-dependent resistor in an electrical circuit), or the underlying state space is not linear but rather a (nontrivial) manifold, the linear Dirac structure has to be replaced by a(n almost) Dirac structure in the standard Courant algebroid on the Pontryagin bundle of this manifold. When consider open systems with ports, whose interface is given by a vector bundle over the state space, then the geometric structure is less understood.~\cite{Merk09} proposed a particular Courant algebroid structure on the bundle $\mathcal{T}M\oplus\mathcal{T}^*M\oplus E\oplus E^*$ induced by a flat connection on $E$ as the geometric structure of port-Hamiltonian systems. However, it is not quite clear why precisely this structure should be the geometric structure of port-Hamiltonian systems, and whether these structures are the only ones. In this note, we propose a systematic way to define the geometric structure motivated by the intuitive understanding of port-Hamiltonian systems as ``in between'' two closed Hamiltonian systems on $M$ and $E$, given by their (almost) Dirac structures in the standard Courant algebroids over $M$ and $E$, respectively. In particular, each Dirac structure (giving a port-Hamiltonian system) should project onto a Dirac structure over $M$ (the underlying Hamiltonian system) and may be recognised as the projection of a Dirac structure over $E$.

The note is organised as follows. First, we recall the well-known notions of Courant algebroid and morphisms of Courant algebroids. We then study classical Courant algebroid morphism, i.e. morphisms given by vector bundle morphisms and show that, under rather restrictive conditions, the pullback of Courant algebroid structures may be defined. This allows us to define the geometric structure of port-Hamiltonian systems as the pullback of the geometric structure of the closed interaction system given by the interconnection of a port-Hamiltonian system with a (virtual) control neighbourhood.

\section{Courant algebroids}

In this section, we recall the well-known theory of Courant algebroids. Very informally speaking, a Courant algebroid is a pseudo-Euclidean vector bundle with an Loday algebra structure on its sections which fulfil certain conditions. A vector bundle is defined as follows.

\begin{Definition}[{see~\cite{Conlon93}}]
A vector bundle $\pi:E\to M$ of rank $m$ is a triple $(E,M,\pi)$, where $E$ and $M$ are smooth manifolds and $\pi\in\mathcal{C}^\infty(E,M)$ so that $E_x := \pi^{-1}(\set{x})$ has a vector space structure for all $x\in M$, together with a family of commutative diagrams
\begin{center}
\begin{tikzcd}
\pi^{-1}(U_i)\arrow[r]\arrow[dr,"\pi"] & U_i\times\mathbb{R}^m\arrow[d,"p_1"]\\
& U_i
\end{tikzcd}
\end{center}
whose arrows are smooth functions restricting to linear functions on fibres $E_x$ with $p_1$ denoting the Cartesian projection onto the first factor, and the $U_i$ are an open covering of $M$. A function $\varphi\in\mathcal{C}^\infty(M,E)$ with $\pi\circ \varphi = \mathrm{id}$ is a (global) \textit{section} of $E$; the set of all (global) sections is denoted with $\Gamma(E)$. If $N$ is a submanifold of $M$ and $S$ is a subbundle of $E\vert_N$ (which is called a subbundle of $E$ supported on $N$), then $\Gamma(E;S):=\set{\varphi\in\Gamma(E)\,\big\vert\,\varphi\vert_N\in\Gamma(S)}$. A \textit{morphism} of vector bundles $\pi:E\to M$ and $\pi':E'\to M'$ is a commutative diagram
\begin{center}
\begin{tikzcd}
E\arrow[r,"\varphi"]\arrow[d,"\pi"] & E'\arrow[d,"\pi'"]\\
M\arrow[r,"\varphi_0"] & M'
\end{tikzcd}
\end{center}
so that $\varphi_x:E_x\to E'_{\varphi_0(x)}$ is linear for all $x\in M$; one calls $\varphi:E\to E'$ a vector bundle morphism over the base $\varphi_0$.\-
\end{Definition}

\begin{Definition}\label{def:Courant_algebroid}[{see~\cite{JoLe18}}]
A \textit{Courant algebroid structure} on the vector bundle $\pi:E\to M$ is triple $(\rho,\langle\cdot,\cdot\rangle,[\![\cdot,\cdot]\!])$ of smooth functions
\begin{align*}
\langle\cdot,\cdot\rangle: &\ E\times_M E \to \mathbb{R},\\[0.2cm]
[\![\cdot,\cdot]\!]: &\ \Gamma(E)\times\Gamma(E) \to \Gamma(E),
\end{align*}
so that $\langle\cdot,\cdot\rangle_x:E_x\times E_x\to \mathbb{R}$ is a pseudo-Euclidean metric for all $x\in M$ and $[\![\cdot,\cdot]\!]$ is bilinear, and a vector bundle morphism $\rho:E\to\mathcal{T}M$ over the identity (the \textit{anchor}) so that, for all $f,g,h\in\Gamma(E)$, $x\in M$ and $e\in E_x$,
\begin{enumerate}
\item[(i)] $[\![f,[\![g,h]\!]]\!] = [\![[\![f,g]\!],h]\!] + [\![g[\![f,h]\!]]\!]$,\vspace{0.2cm}
\item[(ii)] $\rho(f(x))\langle g(\cdot),h(\cdot)\rangle = \langle[\![f,g]\!](x),h(x)\rangle+\langle g(x),[\![f,h]\!](x)\rangle$,\vspace{0.2cm}
\item[(iii)] $\langle e,[\![f,g]\!](x)+[\![g,f]\!](x)\rangle = \rho(e)\left(\langle f(\cdot),g(\cdot)\rangle\right)$.
\end{enumerate}
\end{Definition}

A Courant algebroid is a vector bundle with a fixed Courant algebroid structure. If the specific Courant algebroid structure is not of interest, then we say call the vector bundle Courant algebroid. We give some examples.

\begin{Example}\label{ex:ex_1}
\begin{enumerate}
\item[(i)] Let $M$ be a smooth manifold. The \textit{Pontryagin bundle} of $M$ is the the Whitney sum $\mathcal{T}M\oplus\mathcal{T}^*M$. The \textit{standard Courant-algebroid structure} on $\mathcal{T}M\oplus\mathcal{T}^*M$ consists of the anchor
\begin{align*}
\rho: \mathcal{T}M\oplus\mathcal{T}^*M,\qquad (f,e)\mapsto f,
\end{align*}
the metric
\begin{small}
\begin{align*}
\langle\cdot,\cdot\rangle_+: (\mathcal{T}M\oplus\mathcal{T}^*M)\times_M(\mathcal{T}M\oplus\mathcal{T}^*M) & \to \mathcal{T}M\oplus\mathcal{T}^*M,\\
\big((f,e),(f',e')\big) & \mapsto e(f')+e'(f)
\end{align*}
\end{small}
\noindent and the \textit{Courant-Dorfman bracket}
\begin{small}
\begin{align*}
&[\![\cdot,\cdot]\!]:\Gamma(\mathcal{T}M\oplus\mathcal{T}^*M)\times\Gamma(\mathcal{T}M\oplus\mathcal{T}^*M) \to\Gamma(\mathcal{T}M\oplus\mathcal{T}^*M),\\
&\hspace{1.85cm}\big((X,\alpha),(X',\alpha')\big) \mapsto ([X,X'],\mathfrak{L}_X\alpha'-\iota_{X'}\mathrm{d}\alpha),
\end{align*}
\end{small}
\noindent where $[\cdot,\cdot]$ is the usual Lie bracket of vector fields and
\begin{align*}
\mathfrak{L}_X\alpha' := \iota_{X}\mathrm{d}\alpha'+\mathrm{d}(\iota_X\alpha')
\end{align*}
is the Lie derivative of the differential one form $\alpha'$ along the vector field $X$. This Courant algebroid was studied (albeit with a different, skew-symmetric bracket) by~\cite{Cour90}.
\item[(ii)] Let $(\rho,\langle\cdot,\cdot\rangle,[\![\cdot,\cdot]\!])$ be a Courant algebroid structure on $\pi:E\to M$, and let $\lambda\in\mathbb{R}\setminus\set{0}$. Then $\lambda\langle\cdot,\cdot\rangle$ is also a pseudo-Euclidean metric on $E$ and the conditions (ii) and (iii) are fulfilled due to $\mathbb{R}$-linearity, so that $(\rho,\lambda\langle\cdot,\cdot\rangle,[\![\cdot,\cdot]\!])$ is a Courant algebroid structure on $E$.
\end{enumerate}
\end{Example}

\begin{Remark}
It can be shown that
\begin{align*}
[\![f,\lambda g]\!] = \lambda[\![f,g]\!] + \rho(f)(\lambda)g
\end{align*}
holds for all $f,g\in\Gamma(E)$ and $\lambda\in\mathcal{C}^\infty(M)$, see~\cite{Uchi02}. Therefore, it is straightforward to show that
\begin{equation}\label{eq:Dorfman-Leibniz-rule}
\begin{aligned}
[\![\lambda f,\mu g]\!] & = \lambda\mu[\![f,g]\!] + \lambda\rho(f)(\mu)g -\mu\rho(g)(\lambda)g\\
&\quad +\langle f,g\rangle\mu D_\rho(\lambda)
\end{aligned}
\end{equation}
holds for all $f,g\in\Gamma(E)$ and $\lambda,\mu\in\mathcal{C}^\infty(M)$.
\end{Remark}

Recall that there exists a unique Courant algebroid structure on the Cartesian product of two Courant algebroids $\pi_1:E_1\to M_1$ and $\pi_2:E_2\to M_2$ (with the product bundle structure), which coincides on $E_1\times 0_{E_2}(M_2)$ and $0_{E_1}(M_1)\times E_2$ with the Courant algebroid structures. A Courant algebroid morphism is then defined as follows.

\begin{Definition}[{\cite[Definition 3.1]{Vyso20}}]\label{def:CA_morphism}\ \\
Let $(\rho_1,\langle\cdot,\cdot\rangle_1,[\![\cdot,\cdot]\!]_1)$ and $(\rho_2,\langle\cdot,\cdot\rangle_2,[\![\cdot,\cdot]\!]_2)$ be Courant algebroid structures on $\pi_1:E_1\to M_1$ and $\pi_2:E_2\to M_2$. Let $\overline{E}_2$ denote $\pi_2:E_2\to M_2$ with the Courant algebroid structure $(\rho_2,-\langle\cdot,\cdot\rangle_2,[\![\cdot,\cdot]\!]_2)$. A Courant algebroid morphism is an involutive structure $R\subseteq E_1\times \overline{E}_2$ supported on $S = \mathrm{graph}\,\varphi$ for a smooth function $\varphi:M_1\to M_2$, i.e.
\begin{enumerate}
\item[(i)] $R$ is a subbundle of the restriction of the vector bundle $E_1\times E_2$ to $(\pi_1,\pi_2)^{-1}(\mathrm{graph}\,\varphi) = (E_1\times E_2)\vert_{\mathrm{graph}\,\varphi}$.
\item[(ii)] $\forall (e_1,e_2),(e_1',e_2')\in R: \langle e_1,e_1'\rangle_1 = \langle e_2,e_2'\rangle_2$, i.e. $R$ is isotropic
\item[(iii)] $\forall f,g\in\Gamma(E_1\times E_2;R):[\![f,g]\!]\in\Gamma(E_1\times E_2;R)$, i.e. $R$ is involutive
\item[(iv)] $\forall (e_1,e_2)\in R^{\bot\!\!\!\bot}: \rho(e_1,e_2)\in\mathcal{T}S$, i.e. $R^{\bot\!\!\!\bot}$, the annihilator of $R$ with respect to the product metric, is compatible with the anchor.
\end{enumerate}
If $R$ is the graph of a function $\psi:E_1\to E_2$, then $\psi$ is called \textit{classical Courant algebroid morphism}.
\end{Definition}

\section{Classical Courant algebroid morphism}

In this section, we study classical Courant algebroid morphisms in greater detail. Evidently, each classical Courant algebroid morphism is necessarily a vector bundle morphism. The simplest case is a vector bundle morphism over the identity. In that case, we derive the following characterisation.

\begin{Lemma}\label{lem:CA_morphism_over_identity}
Let $\pi_1:E_1\to M$ and $\pi_2:E_2\to M$ be vector bundles over the same manifold $M$, and let $(\rho_1,\langle\cdot,\cdot\rangle_1,[\![\cdot,\cdot]\!]_1)$ and $(\rho_2,\langle\cdot,\cdot\rangle_2,[\![\cdot,\cdot]\!]_2)$ be Courant algebroid structures on $E_1$ and $E_2$, respectively. Let further $\varphi:E_1\to E_2$ be a vector bundle morphism over the identity on $M$. Then $\varphi$ is a classical Courant algebroid morphism if, and only if,
\begin{align}\label{eq:function_CA_morphism}
\forall f,g\in\Gamma(E_1): \varphi\circ [\![f,g]\!]_1 = [\![\varphi\circ f,\varphi\circ g]\!]_2,
\end{align}
\begin{align}\label{eq:function_CA_morphism_2}
\forall e,e'\in E_1: \langle e,e'\rangle_1 = \langle\varphi(e),\varphi(e')\rangle_2,
\end{align}
and
\begin{align}\label{eq:function_CA_morphism_3}
\forall e\in E_1: \rho_1(e) = \rho_2(\varphi(e)).
\end{align}
\end{Lemma}
\begin{proof}
The isotropy condition (ii) of Definition~\ref{def:CA_morphism} is evidently equivalent to~\eqref{eq:function_CA_morphism_2}. It remains to show the equivalence of the conditions~\eqref{eq:function_CA_morphism} and~\eqref{eq:function_CA_morphism_3} and the conditions (iii) and (iv) of Definition~\ref{def:CA_morphism}.

``$\implies$'' Let $\varphi$ be a classical Courant algebroid morphism.  We show~\eqref{eq:function_CA_morphism}. Let $f,g\in\Gamma(E_1)$. Since $\varphi$ is a vector bundle morphism over the identity, $\varphi\circ f,\varphi\circ g\in\Gamma(E_2)$. Put $\widehat{f} := (p_1^*f,p_2^*(\varphi\circ f))$ and $\widehat{g} := (p_1^*g,p_2^*(\varphi\circ g))$. By construction, we have $\widehat{f},\widehat{g}\in\Gamma(E_1\times E_2;\mathrm{graph}\,\varphi)$. Hence, involutivity of $\mathrm{graph}\,\varphi$ yields in view of the definition of the Courant-Dorfman bracket for $E_1\times \overline{E_2}$ that
\begin{align*}
[\![\widehat{f},\widehat{g}]\!](x,x) = ([\![f,g]\!]_1(x),[\![\varphi\circ f,\varphi\circ g]\!]_2(x))\in\mathrm{graph}\,\varphi,
\end{align*}
and hence $\varphi\circ [\![f,g]\!]_1(x) = [\![\varphi\circ f,\varphi\circ g]\!]_2(x)$ for all $x\in M$.
This shows that $\varphi$ fulfils~\eqref{eq:function_CA_morphism}. It remains to show~\eqref{eq:function_CA_morphism_3}. Unless the rank of $E_2$ is zero, we have $\mathrm{graph}\,\varphi\neq E_1\times E_2\vert_{\delta_M(M)}$, where $\delta_M:M\hookrightarrow M\times M$ denotes the diagonal embedding so that $\delta_M(M) = \mathrm{graph}\,\mathrm{id}_M$, and hence~\cite[Proposition 2.16]{Vyso20} yields that $\mathrm{graph}\,\varphi$ is compatible with the anchor $\rho$. In the former case, however, injectivity of $\varphi$ (which holds by~\eqref{eq:function_CA_morphism}) implies the rank of $E_1$ to be zero and hence $\rho = 0$, which is compatible which every subvector bundle supported on any submanifold. Hence~\cite[Lemma 4.2]{Vyso20} implies~\eqref{eq:function_CA_morphism_3}.

``$\impliedby$'' We show involutivity of $\Gamma(E_1\times E_2;\mathrm{graph}\,\varphi)$. Let $f,f'\in\Gamma(E_1\times E_2;\mathrm{graph}\,\varphi)$. Recall that $p_1^*\Gamma(E_1)\oplus p_2^*\Gamma(E_2)$ is a generating set for $\Gamma(E_1\times E_2)$ over $\mathcal{C}^\infty(M\times M)$, i.e. there are $\lambda_i,\lambda_j'\in\mathcal{C}^\infty(M\times M)$, $f_i,f_j'\in\Gamma(E_1)$ and $g_i,g_j'\in\Gamma(E_2)$, $i,j\in\set{1,\ldots,n}$, so that
\begin{align*}
f = \sum_{i = 1}^n\lambda_i \left(p_1^*f_i\oplus p_2^* g_i\right),\quad\text{and}\quad f' = \sum_{j = 1}^n\lambda_j' \left(p_1^*f_j'\oplus p_2^* g_j'\right).
\end{align*}
Consider the sections
\begin{align*}
\widehat{f} & = \sum_{i = 1}^n\lambda_i \left(p_1^*f_i\oplus p_2^* \varphi\circ f_i\right),\\
\widehat{f}' & = \sum_{j = 1}^n\lambda_j' \left(p_1^*f_j'\oplus p_2^* \varphi\circ f_j'\right),
\end{align*}
which are certainly contained in $\Gamma(E_1\times E_2;\mathrm{graph}\,\varphi)$ and fulfil, by linearity of $\varphi$,
\begin{align*}
\left(\widehat{f}-f\right)\vert_{\delta_M(M)} = \left(\widehat{f}'-f'\right)\vert_{\delta_M(M)} = 0.
\end{align*}
In view of~\cite[Proposition 2.12]{Vyso20}, we conclude
\begin{align*}
[\![\widehat{f},f']\!]\vert
_{\delta_M(M)} - [\![f,f']\!]\vert
_{\delta_M(M)} = [\![\widehat{f}-f,f']\!]\vert
_{\delta_M(M)} = 0
\end{align*}
and therefore (by exchanging $f$ and $f'$, and $\widehat{f}$ and $\widehat{f}'$, respectively)
\begin{align*}
[\![f,f']\!]\vert
_{\delta_M(M)} = [\![\widehat{f},\widehat{f}']\!]\vert
_{\delta_M(M)}.
\end{align*}
Then,~\eqref{eq:Dorfman-Leibniz-rule} implies on $\delta_M(M)$
\begin{align*}
[\![\widehat{f},\widehat{f}']\!] & = \sum_{i,j = 1}^n\Big( \lambda_i\lambda_j')\left([\![f_i,f_j']\!]_1,[\![\varphi\circ f_i,\varphi\circ f_j']\!]_2\right)\\
& \qquad\quad + \lambda_i\rho(p_1^*f_i\oplus p_2^* \varphi\circ f_i)(\lambda_j')(f_j',\varphi\circ f_j')\\
& \qquad\quad -\lambda_j'\rho(p_1^*f_j'\oplus p_2^* \varphi\circ f_j')(\lambda_i)(f_i,\varphi\circ f_i)\Big)\\
& = \sum_{i,j = 1}^n\Big( \lambda_i\lambda_j'\left([\![f_i,f_j']\!]_1,\varphi\circ [\![f_i,f_j']\!]_1\right)\\
& \qquad\quad + \lambda_i\rho(p_1^*f_i\oplus p_2^* \varphi\circ f_i)(\lambda_j')(f_j',\varphi\circ f_j')\\
& \qquad\quad -\lambda_j'\rho(p_1^*f_j'\oplus p_2^* \varphi\circ f_j')(\lambda_i)(f_i,\varphi\circ f_i)\Big)
\end{align*}
and hence $[\![\widehat{f},\widehat{f}']\!]\in\Gamma(E_1\times E_2;\mathrm{graph}\,\varphi)$. This shows $[\![f,f']\!]\in\Gamma(E_1\times E_2;\mathrm{graph}\,\varphi)$ and therefore $\mathrm{graph}\,\varphi$ is indeed involutive. Unless $\mathrm{rk}\, E_2 = 0$, this implies in view of~\cite[Proposition 2.16]{Vyso20} that $(\mathrm{graph}\,\varphi)^{\bot\!\!\!\bot}$ is compatible with the anchor. If, however, $\mathrm{rk}\, E_2 = 0$, then $\mathrm{rk}\, E_1 = 0$ and hence $(\mathrm{graph}\,\varphi)^{\bot\!\!\!\bot} = \mathrm{graph}\,\varphi$ is compatible with the then trivial anchor.
\end{proof}

Lemma~\ref{lem:CA_morphism_over_identity} makes heavy use of the fact that the image of a section under a vector bundle morphism over the identity is again a section. In general, this is no longer true, unless the base of the vector bundle morphism is a diffeomorphism. Therefore, the notion of $\varphi$-related sections has been introduced, see~\cite{HiggMack90}.

\begin{Definition}
Let $\pi:E\to M$ and $\tau:F\to N$ be vector bundles and $(\varphi,\varphi_0):(E,M)\to (F,N)$ a vector bundle morphism. Sections $f\in\Gamma(E)$ and $g\in\Gamma(F)$ are \textit{$\varphi$-related} if, and only if, $\varphi\circ f = g\circ\varphi_0$; we write $f\sim_{\varphi} g$.
\end{Definition}

Of course, given a section $f\in\Gamma(E)$, the existence of a $\varphi$-related section of $F$ is not guaranteed as the following examples illustrate.

\begin{Example}
\begin{enumerate}
\item[(i)] If $\varphi_0$ is not injective, then a necessary condition for the existence of a $\varphi$-related section to $f\in\Gamma(E)$ is that $f(\varphi_0^{-1}(\set{x}))$ contains at most one element for each $x\in N$, where we put $f(\emptyset) = \emptyset$.
\item[(ii)] Consider the trivial line bundles 
\begin{align*}
\pi:\mathbb{R}^{1+1} & \to\mathbb{R},\qquad (x,z)\mapsto x,\\
\tau: \mathbb{R}^{2+1} & \to\mathbb{R}^2,\qquad (x,y,z)\mapsto (x,y)
\end{align*}
with vector bundle morphism 
\begin{align*}
\varphi:\mathbb{R}^{1+1}\to\mathbb{R}^{2+1}, (x,z)\mapsto (x^2,x^3,z).
\end{align*}
The base morphism $\varphi_0 = x\mapsto (x^2,x^3)$ is certainly differentiable, but not an immersion. Consider the smooth section
\begin{align*}
f:\mathbb{R}\to\mathbb{R}^{1+1},\qquad x\mapsto (x,x)
\end{align*}
of $\mathbb{R}^{1+1}$. Let $g:\mathbb{R}^{2}\to\mathbb{R}^{2+1}$ be a function so that $\varphi\circ f = g\circ\varphi_0$, i.e.
\begin{align*}
\forall x\in\mathbb{R}, g(x^2,x^3) = x.
\end{align*}
Seeking a contradiction, assume that $g$ is differentiable. Then, however, the function
\begin{align*}
h: \mathbb{R}\to\mathbb{R},\qquad y\mapsto g(y,y^{\frac{3}{2}})
\end{align*}
is differentiable and an inverse function of $x\mapsto x^2$; therefore, $h = \sqrt{\cdot}$, which is not differentiable in zero. In particular, this shows that there is no section of $\mathbb{R}^{2+1}$ that is $\varphi$-related to $f$.
\item[(iii)] Consider the trivial line bundles
\begin{align*}
\pi:\mathbb{R}^{1+1} & \to\mathbb{R},\qquad (x,y)\mapsto x,\\
\pi:(-\infty,0)\times\mathbb{R} & \to (-\infty,0),\qquad (x,y)\mapsto x
\end{align*}
and the vector bundle morphism
\begin{align*}
\varphi: (-\infty,0)\times\mathbb{R}\to\mathbb{R}^{1+1},\qquad (x,y)\mapsto (x,y).
\end{align*}
The base morphism $\varphi_0 = x\mapsto x$ is clearly an immersion. Consider the smooth section 
\begin{align*}
f:(-\infty,0)\to (-\infty,0)\times\mathbb{R},\qquad x\mapsto \left(x,\frac{1}{x}\right)
\end{align*}
of $(-\infty,0)\times \mathbb{R}$. Since there exists no smooth function $g:\mathbb{R}\to\mathbb{R}$ so that $g(x) = \frac{1}{x}$ for all $x\in (-\infty,0)$, there exists no smooth section $g\in\Gamma(\mathbb{R}^{1+1})$ so that $f$ and $g$ are $\varphi$-related.
\end{enumerate}
\end{Example}

When we assume the existence of $\varphi$-related sections (which is, as we have seen, not the case in general), then a classical Courant algebroid morphism may be characterised as follows.

\begin{Proposition}\label{prop:char_with_phi_rel_sec}
Let $(\rho_1,\langle\cdot,\cdot\rangle_1,[\![\cdot,\cdot]\!]_1)$ and $(\rho_2,\langle\cdot,\cdot\rangle_2,[\![\cdot,\cdot]\!]_2)$ be Courant algebroid structures on $\pi_1:E_1\to M_1$ and $\pi_2:E_2\to M_2$, resp., and let $\varphi:E_1\to E_2$ be a vector bundle morphism over $\varphi_0$ so that, for each $f\in\Gamma(E_1)$ there exists a $\varphi$-related section $g\in\Gamma(E_2)$. Then, $\varphi$ is a classical Courant algebroid morphism if, and only if, for each $f_1,f_2\in\Gamma(E_1)$ and any (and hence each) $g_1,g_2\in\Gamma(E_2)$ with $f_1\sim_\varphi g_1$ and $f_2\sim_\varphi g_2$,
\begin{equation}\label{eq:CA_morphism_function_bracket}
\varphi\circ[\![f_1,f_2]\!]_1 = [\![g_1,g_2]\!]_2\circ\varphi_0
\end{equation}
and
\begin{align}\label{eq:CA_morphism_function_metric}
\langle f_1,f_2\rangle_1 = \langle g_{1},g_{2}\rangle_2\circ\varphi_0
\end{align}
and
\begin{align}\label{eq:CA_morphism_function_anchor}
\rho_2\circ \varphi = \mathrm{d}\varphi_0\circ\rho_1.
\end{align}
\end{Proposition}
\begin{proof}
``$\implies$'' Let $\varphi$ be a classical Courant algebroid morphism. Let $f_1,f_2$ and consider $\varphi$-related sections $g_1,g_2\in\Gamma(E_2)$. Then, isotropy of $\mathrm{graph}\,\varphi$ implies
\begin{align*}
\langle f_1,f_2\rangle_1 & = \langle\varphi\circ f_1,\varphi\circ f_2\rangle_2 = \langle g_1,g_2\rangle_2\circ\varphi_0
\end{align*}
and hence~\eqref{eq:CA_morphism_function_metric} holds true. Recall that $\Gamma(E_1\times E_2;\mathrm{graph}\,\varphi)$ denotes the sections of $E_1\times E_2$ which take on $\mathrm{graph}\,\varphi_0$ values in $\mathrm{graph}\,\varphi$. In particular, $p_1^*f_i\oplus p_2^*g_i\in\Gamma(E_1\times E_2;\mathrm{graph}\,\varphi)$, $i \in\set{1,2}$. Since $\varphi$ is a classical Courant algebroid morphism, $\Gamma(E_1\times E_2;\mathrm{graph}\,\varphi)$ is involutive and hence
\begin{equation}\label{eq:irrsinnig_lang}
\left[\!\left[p_1^*f_1\oplus p_2^*g_1,p_1^*f_2\oplus p_2^*g_2\right]\!\right] = p_1^*[\![f_1,f_2]\!]_1\oplus p_2^*[\![g_1,g_2]\!]
\end{equation}
is contained in $\Gamma(E_1\times E_2;\mathrm{graph}\,\varphi)$. Evaluating~\eqref{eq:irrsinnig_lang} on $\mathrm{graph}\,\varphi_0$ yields~\eqref{eq:CA_morphism_function_bracket}. The identity~\eqref{eq:CA_morphism_function_anchor} is a consequence of~\cite[Proposition 2.16]{Vyso20}.

``$\impliedby$'' Let $(\varphi,\varphi_0): (E_1,M_1)\to (E_2,M_2)$ be a vector bundle morphism with~\eqref{eq:CA_morphism_function_bracket}--\eqref{eq:CA_morphism_function_anchor}. We show that $\varphi$ is a classical Courant algebroid morphism. First, we show that $\mathrm{graph}\,\varphi$ is isotropic. Let $e,e'\in E_1$ with $\pi(e_1) = \pi(e_2) =: x$. Then, there exist sections $f_1,f_2\in\Gamma(E_1)$ with $f_1(x) = e_1$ and $f_2(x) = e_2$. Let $g_1,g_2\in\Gamma(E_2)$ with $f_1\sim_\varphi g_1$ and $f_2\sim_\varphi g_2$, so that~\eqref{eq:CA_morphism_function_metric} holds. Then, we have
\begin{align*}
0 & = \langle f_1(x),f_2(x)\rangle_1 - \left\langle g_1(\varphi_0(x)),g_2(\varphi_0(x))\right\rangle_2\\
& = \langle e_1,e_2\rangle_1-\langle \varphi(e_1),\varphi(e_2)\rangle_2 = \langle (e_1,\varphi(e_1)),(e_2,\varphi(e_2))\rangle.
\end{align*}
This shows that $\mathrm{graph}\,\varphi$ is isotropic. A straightforward calculation yields from~\eqref{eq:CA_morphism_function_anchor} that $\mathrm{graph}\,\varphi$ is compatible with the anchor. We show that $\Gamma(E_1\times E_2;\mathrm{graph}\,\varphi)$ is involutive. Let $h_1,h_2\in \Gamma(E_1\times E_2;\mathrm{graph}\,\varphi)$ and define
\begin{align*}
f_i: M_1\to E_1,\qquad x\mapsto p_1(h_i(x,\varphi_0(x))),\qquad i\in\set{1,2}.
\end{align*}
Since $h_1$ and $h_2$ are sections of the product vector bundle $E_1\times E_2$, we have $f_1,f_2\in\Gamma(E_1)$. Let $g_1,g_2\in\Gamma(E_2)$ be $\varphi$-related to $f_1$ and $f_2$, respectively. Then, $p_1^*f_i\oplus p_2^*g_i\in\Gamma(E_1\times E_2;\mathrm{graph}\,\varphi)$ and we have
\begin{align*}
\left.\left(h_i-p_1^*f_i\oplus p_2^*g_i\right)\right\vert_{\mathrm{graph}\,\varphi_0}\equiv 0,\quad i\in\set{1,2}.
\end{align*}
Therefore,~\cite[Proposition 2.12]{Vyso20} implies 
\begin{align*}
[\![h_1,h_2]\!] & = \left[\!\left[h_1,p_1^*f_2\oplus p_2^*g_2\right]\!\right]\\
& = \left[\!\left[p_1^*f_1\oplus p_2^*g_1,p_1^*f_2\oplus p_2^*g_2\right]\!\right].
\end{align*} 
Then~\eqref{eq:irrsinnig_lang} yields with~\eqref{eq:CA_morphism_function_bracket} that $[\![h_1,h_2]\!]\in\Gamma(E_1\times E_2;\mathrm{graph}\,\varphi)$. This shows that $\Gamma(E_1\times E_2;\mathrm{graph}\,\varphi)$ is isotropic; and an analogous argument yields that we may replace $g_1$ and $g_2$ by any other $\varphi$-related sections. Finally, compatibility of $\left(\mathrm{graph}\,\varphi\right)^{\bot\!\!\!\bot}$ with the anchor is a direct consequence of~\cite[Proposition 2.16]{Vyso20}.
\end{proof}

\begin{Remark}
By using $\varphi$-representations (see~\cite{HiggMack90}), this characterisation may be generalised to the case that the existence of $\varphi$-related sections cannot be guaranteed.
\end{Remark}

If the morphism is injective and its base is an immersion onto a closed regular submanifold, then we can show that the existence of $\varphi$-related sections is guaranteed.

\begin{Lemma}\label{lem:phi_related_existence}
Let $\pi:E\to M$ and $\tau:F\to N$ be vector bundles and $(\varphi,\varphi_0):(E,M)\to (F,N)$ a vector bundle morphism so that $\varphi_0$ is an immersion and $\varphi_0(M)$ is a closed regular submanifold of $N$. Then, for each $f\in\Gamma(E)$ there exists $g\in\Gamma(F)$ so that $f$ and $g$ are $\varphi$-related.
\end{Lemma}
\begin{proof}
Let $f\in\Gamma(E)$. Define the function $f_0 := \varphi\circ f\circ\varphi_0^{-1}$. Since $(\varphi,\varphi_0)$ is a vector bundle morphism and since $f$ is a section of $E$, $f_0$ is a section of the restricted vector bundle $F\vert_{\varphi_0(M)}$. Since $\varphi_0(M)$ is, by assumption, a closed regular submanifold of $N$ a partition of unity argument implies the existence of a section $g\in\Gamma(F)$ so that $g\vert_{\varphi_0(M)} = f_0$. In particular, we have
\begin{align}\label{eq:Dorfman_assumption}
g\circ\varphi_0 = f_0\circ\varphi_0 = \varphi\circ f\circ\varphi_0^{-1}\circ\varphi_0 = \varphi\circ f.
\end{align}
This shows that $f$ and $g$ are $\varphi$-related.
\end{proof}

Under the assumptions of Lemma~\ref{lem:phi_related_existence}, we may show that each such vector bundle morphism into a Courant algebroid generates a unique pullback Courant algebroid structure on the source of the morphism.

\begin{Proposition}\label{prop:pullback_algebroid}
Let $(E,M,\pi,\rho,\langle\cdot,\cdot\rangle,[\![\cdot,\cdot]\!])$ be a Courant algebroid, $\pi':E'\to M'$ a vector bundle and $(\varphi,\varphi_0):(E',M')\to (E,M)$ an injective vector bundle morphism so that $\varphi_0(M')$ is a closed, regular submanifold of $M$ and $\rho(\varphi_0(M'))\subseteq\mathcal{T}\varphi_0(M')$. There exists a Courant algebroid structure $(\rho',\langle\cdot,\cdot\rangle',[\![\cdot,\cdot]\!]')$ on $E'$ so that $\varphi$ is a Courant algebroid morphism if, and only if, $\langle\cdot,\cdot\rangle$ is fiberwise nondegenerate on $\mathrm{im}\,\varphi$ and $\Gamma(E;\varphi(E'))$ is involutive. If it exists, this Courant algebroid structure is unique.
\end{Proposition}
\begin{proof}
``$\implies$'' Let $(\rho',\langle\cdot,\cdot\rangle',[\![\cdot,\cdot]\!]$ be a Courant algebroid structure on $E'$ so that $\varphi$ is a classical Courant algebroid morphism. In particular, $\mathrm{graph}\,\varphi$ is isotropic with respect to the fiberwise inner product
\begin{align*}
\langle\!\langle\cdot,\cdot\rangle\!\rangle: (E'\times E)\times_{M'\times M}(E'\times E) & \to \mathbb{R},\\
\big(x',x),(y',y)\big) & \mapsto \langle x',y'\rangle'-\langle x,y\rangle.
\end{align*}
Equivalently, we have
\begin{align}\label{eq:isotropy_1}
\forall (x',y')\in E'\times_{M'}E': \langle x',y'\rangle' = \langle\varphi(x'),\varphi(y')\rangle.
\end{align}
Let $x_0\in E'$ so that the restriction of $\langle\varphi(x_0),\cdot\rangle$ vanishes on $\mathrm{im}\varphi\cap\pi^{-1}(\pi(\varphi(x_0)))$. Since $\varphi$ is a vector bundle morphism, we have 
\begin{align*}
\mathrm{im}\varphi\cap\pi^{-1}(\pi(\varphi(x_0))) & = \set{z\in\mathrm{im}\,\varphi\,\big\vert\,\pi(z) = \pi(\varphi(x_0))}\\
& = \set{\varphi(x')\,\big\vert\,\pi'(x') = \pi'(x_0)}.
\end{align*}
Therefore,~\eqref{eq:isotropy_1} yields
\begin{align*}
\forall x'\in\pi'^{-1}(\pi'(x_0)): \langle x_0,x'\rangle' = \langle \varphi(x_0),\varphi(x')\rangle = 0
\end{align*}
and thus $\langle x_0,\cdot\rangle' \equiv 0$. Since $\langle \cdot,\cdot\rangle$ is nondegenerate, we conclude that $x_0 = 0_{E'}(\pi'(x_0))$; and since $\varphi$ is fiberwise linear, this yields $\varphi(x_0) = 0_{E}(\pi(\varphi(x_0)))$. This shows that $\langle\cdot,\cdot\rangle$ is fiberwise nondegenerate on $\mathrm{im}\,\varphi$. Let $f,g\in\Gamma(E;\varphi(E'))$. Since $\varphi$ is (as classical Courant algebroid morphism) injective, the sections $\widetilde{f} := \varphi^{-1}\circ f\circ \varphi_0$ and $\widetilde{g} := \varphi^{-1}\circ g\circ\varphi_0$ are well-defined and $\varphi$-related to $f$ and $g$, respectively. Since $\varphi$ is, by assumption, a Courant algebroid morphism, we conclude that
\begin{align*}
[\![f,g]\!]\circ\varphi_0 = \varphi\circ[\![\widetilde{f},\widetilde{g}]\!]'\in\Gamma \varphi_0^*\varphi(E')
\end{align*}
and hence $[\![f,g]\!]\in\Gamma(E;\varphi(E'))$.

``$\impliedby$'' By assumption, $\varphi_0$ is an immersion onto its image $\varphi_0(M')$. In particular, $\varphi_0$ is a bijection onto its image and admits an inverse function $\psi:\varphi_0(M)\to M'$. From the implicit function theorem, it follows that $\psi$ is a smooth function. Define the anchor
\begin{align*}
\rho':E'\to\mathcal{T}M',\quad e\mapsto(\mathrm{d}\psi\circ\rho\circ\varphi)(e).
\end{align*}
By assumption, $\rho\circ\varphi$ maps into $\mathcal{T}\varphi_0(M')$ and hence $\rho'$ is well-defined. Since $\rho$ is a vector bundle morphism over the identity, and $\mathrm{d}\psi$ is a vector bundle morphism over $\psi$, $\rho'$ is a vector bundle morphism over the identity. Since $\langle\cdot,\cdot\rangle$ is a fiberwise bilinear form on $E$ and since $\varphi$ is a vector bundle morphism, the function
\begin{align*}
\langle\cdot,\cdot\rangle': E'\oplus_{M'} E'\to \mathbb{R},\qquad (e,e')\mapsto \langle\varphi(e),\varphi(e')\rangle.
\end{align*}
is a fiberwise bilinear form on $E'$. Let $e\in E'$ so that $\langle e,\cdot\rangle'\equiv 0$; by construction this is equivalent to $\langle\varphi(e),\cdot\rangle\vert_{\varphi(\pi'^{-1}(e))}\equiv 0$. Since the restriction of $\langle\cdot,\cdot\rangle$ to $\mathrm{im}\varphi$ is nondegenerate by assumption, we conclude that $\varphi(e) = 0$; since $\varphi$ is injective, we have $e = 0_{E'}(\pi'(e))$. This shows that $\langle\cdot,\cdot\rangle'$ is a fiberwise nondegenerate inner product on $E'$. Lastly, we construct a Dorfman bracket on $E'$. Let $f,g\in\Gamma(E')$. In view of Lemma~\ref{lem:phi_related_existence}, there are sections $\widehat{f},\widehat{g}\in\Gamma(E)$ so that $\varphi\circ f = \widehat{f}\circ\varphi_0$, and $\varphi\circ g = \widehat{g}\circ\varphi_0$. Define
\begin{align*}
[\![f,g]\!]' := \varphi^{-1}\circ[\![\widehat{f},\widehat{g}]\!]\circ\varphi_0.
\end{align*}
We show that $[\![f,g]\!]'$ is well-defined. Let $\widetilde{f},\widetilde{g}\in\Gamma(E)$ with $\varphi\circ f = \widetilde{f}\circ\varphi_0$ and $\varphi\circ g = \widetilde{g}\circ\varphi_0$. Let $x\in M'$. Then, there exists an open neighbourhood $U$ of $x$ and $e_1,\ldots,e_k\in\Gamma(E)$ so that $(e_i\vert_U)_{i = 1}^n$ is a frame of the restriction $E\vert_U := \pi^{-1}(U)$. Let $f_1,\ldots,f_n\in\mathcal{C}^\infty(M)$ and $g_1,\ldots,g_n\in\mathcal{C}^\infty(M)$ so that the restrićtions $(f_i\vert_U)_{i = 1}^n$ and $(g_i\vert_U)_{i = 1}^n$ are representations of $(\widetilde{f}-\widehat{f})\vert_{U}$ and $(\widetilde{g}-\widehat{g})\vert_{U}$ with respect to the frame $(e_i\vert_U)_{i = 1}^n$, respectively. By assumption, $\widehat{g}\circ\varphi_0 = \widetilde{g}\circ\varphi_0$ and $\widehat{f} = \widetilde{f}$, and thus $(f_i\circ\varphi)\vert_{\psi(U)} = (g_i\circ\varphi)\vert_{\psi(U)}\equiv 0$ for all $i\in\underline{n}$. In view of~\eqref{eq:Dorfman-Leibniz-rule}, we have, for all $z\in\psi(U)$,
\begin{align*}
[\![\widetilde{f},\widetilde{g}-\widehat{g}]\!](\varphi_0(z)) & = \left[\!\left[h,\sum_{i = 1}^ng_i e_i\right]\!\right](\varphi_0(z))\\
& = \sum_{i = 1}^n \underbrace{g_i(\varphi_0(z))}_{ = 0}[\![\widetilde{f},e_i]\!](\varphi_0(z))\\
&\quad + \underbrace{\rho(\widetilde{f}(\varphi_0(z)))(g_i)(\varphi_0(z))}_{ = 0} e_i(\varphi_0(z)) = 0,
\end{align*}
where $\rho(\widetilde{f}(\varphi_0(z))(g_i)(\varphi_0(z)) = 0$ since $\rho(\widetilde{f}(\varphi_0(z))\in\mathcal{T}\varphi_0(M)$. In particular, we have $[\![\widetilde{f},\widetilde{g}]\!]\circ\varphi_0 = [\![\widetilde{f},\widehat{g}]\!]\circ\varphi_0$; analogously, it may be proven that $[\![\widehat{f},\widetilde{g}]\!] = [\![\widehat{f},\widehat{g}]\!]$ and $[\![\widetilde{g},\widetilde{f}-\widehat{f}]\!] = 0$. By property (iii) in the definition of the Courant algebroid, we have
\begin{align*}
& [\![\widetilde{f}-\widehat{f},\widetilde{g}]\!](\varphi_0(z))\\
&\quad = -[\![\widetilde{g},\widetilde{f}-\widehat{f}]\!](\varphi_0(z)) + D_\rho(\langle \widetilde{f}-\widehat{f},\widetilde{g}\rangle)(\varphi_0(z))\\
&\quad = D_\rho(\langle \widetilde{f}-\widehat{f},\widetilde{g}\rangle)(\varphi_0(z)),
\end{align*}
where $D_\rho(s)\in\Gamma(E)$ denotes the unique section with $\langle D_{\rho}(s),s'\rangle = \rho(s)(s')$ for all $s'\in\Gamma(E)$. In particular, the assumption~\eqref{eq:Dorfman_assumption} yields that $D_\rho(\langle \widetilde{f}-\widehat{f},\widetilde{g}\rangle)(\varphi_0(z))\in\varphi(E')$ and we have for all $e'\in \pi'^{-1}(\set{z})$,
\begin{align*}
0 = \rho(\varphi(e'))(\langle\widetilde{f}-\widehat{f},\widetilde{g}\rangle) = \langle D_\rho(\langle \widetilde{f}-\widehat{f},\widetilde{g}\rangle)(\varphi_0(z)),\varphi(e')\rangle;
\end{align*}
hence $D_\rho(\langle \widetilde{f}-\widehat{f},\widetilde{g}\rangle)(\varphi_0(z))$ is orthogonal on $\varphi(\pi'^{-1}(\set{z})$. Since we assume that $\langle\cdot,\cdot\rangle$ is nondegenerate on $\varphi(E')$, we conclude that $D_\rho(\langle \widetilde{f}-\widehat{f},\widetilde{g}\rangle)(\varphi_0(z)) = 0_{E}(\varphi_0(z))$. Thus, we conclude that $[\![\widehat{f},\widetilde{g}]\!]\circ\varphi_0 = [\![\widetilde{f},\widetilde{g}]\!]\circ\varphi_0$; analogously, we see that $[\![\widehat{f},\widehat{g}]\!]\circ\varphi_0 = [\![\widetilde{f},\widehat{g}]\!]\circ\varphi_0$. Therefore, we have
\begin{align*}
[\![\widehat{f},\widehat{g}]\!]\circ\varphi_0 = [\![\widetilde{f},\widehat{g}]\!]\circ\varphi_0 = [\![\widetilde{f},\widetilde{g}]\!]\circ\varphi_0
\end{align*}
and $[\![f,g]\!]'$ is indeed well-defined. It remains to show that $(\rho',\langle\cdot,\cdot\rangle',[\![\cdot,\cdot]\!]')$ is a Courant algebroid structure on $E'$. Let $f,g,h\in\Gamma(E')$ with $\varphi$-related sections $\widehat{f},\widehat{g},\widehat{h}\in\Gamma(E)$, i.e.
\begin{align*}
\varphi\circ f = \widehat{f}\circ\varphi_0,\quad \varphi\circ g = \widehat{g}\circ\varphi_0,\quad\varphi\circ h = \widehat{h}\circ\varphi_0.
\end{align*}
By construction of $[\![\cdot,\cdot]\!]$, we find that $[\![g,h]\!]', [\![f,g]\!]', [\![f,h]\!]'$ and $[\![\widehat{g},\widehat{h}]\!], [\![\widehat{f},\widehat{g}]\!], [\![\widehat{f},\widehat{h}]\!]$ are $\varphi$-related, respectively. Therefore, we get
\begin{align*}
\varphi\circ[\![f,[\![g,h]\!]']\!]' & = [\![\widehat{f},[\![\widehat{g},\widehat{h}]\!]]\!]\circ\varphi_0\\
& = [\![[\![\widehat{f},\widehat{g}]\!],\widehat{h}]\!]\circ\varphi_0 + [\![\widehat{g},[\![\widehat{f},\widehat{h}]\!]]\!]\circ\varphi_0\\
& = \varphi\circ[\![[\![f,g]\!]',h]\!]'+\varphi\circ[\![g,[\![f,h]\!]']\!]'\\
& = \varphi\circ ([\![[\![f,g]\!]',h]\!]'+[\![g,[\![f,h]\!]']\!]').
\end{align*}
Since $\varphi$ is injective, we conclude that
\begin{align*}
[\![f,[\![g,h]\!]']\!]' = [\![[\![f,g]\!]',h]\!]'+[\![g,[\![f,h]\!]']\!]'
\end{align*}
and thus $(\Gamma(E'),[\![\cdot,\cdot]\!]')$ is a left-Loday algebra. For each $x\in M'$, we have
\begin{align*}
\rho'(f(x))\langle g,h\rangle' & = (\mathrm{d}\psi\circ\rho)((\varphi\circ f)(x))\langle \varphi\circ g,\varphi\circ h\rangle\\
& = \rho((\varphi\circ f)(x))\langle \varphi\circ g\circ\psi,\varphi\circ h\circ \psi\rangle\\
& = \rho(\widehat{f}(\varphi_0(x)))\langle\widehat{g},\widehat{h}\rangle\\
& = \langle [\![\widehat{f},\widehat{g}]\!](\varphi_0(x)),\widehat{h}(\varphi_0(x))\rangle\\
&\quad + \langle \widehat{g}(\varphi_0(x)),[\![\widehat{f},\widehat{g}]\!](\varphi_0(x))\rangle\\
& = \langle [\![f,g]\!](x),h(x)\rangle + \langle g(x),[\![f,h]\!](x)\rangle.
\end{align*}
This shows that $(\rho',\langle\cdot,\cdot\rangle',[\![\cdot,\cdot]\!])$ fulfills the property (ii) in Definition~\ref{def:Courant_algebroid}. It remains to verify that $(\rho',\langle\cdot,\cdot\rangle',[\![\cdot,\cdot]\!])$ fulfills (iii). Let $\alpha\in\mathcal{C}^\infty(M')$ and $\widehat{\alpha}\in\mathcal{C}^\infty(M)$ so that $\alpha = \widehat{\alpha}\circ\varphi_0$. Equivalently, we have $\widehat{\alpha}\vert_{\varphi_0(M')} = \alpha\circ\psi$. Recall that $D_{\rho'}(\alpha)\in\Gamma(E')$ is the unique solution of
\begin{align*}
\langle D_{\rho'}(\alpha)(x),k(x)\rangle' = \rho'(k(x))(\alpha)
\end{align*}
for all $x\in M'$ and $k\in\Gamma(E')$. Let $x\in M'$ and $k\in\Gamma(E')$ with $\varphi$-related section $\widehat{k}\in\Gamma(E)$. Let further $\ell\in\Gamma(E')$ be $\varphi$-related to $D_{\rho'}(\alpha)$. Then, we have
\begin{align*}
\langle\ell(\varphi_0(x)),\widehat{k}(\varphi_0(x))\rangle & = \langle\varphi(D_{\rho'}(\alpha)(x)),\varphi(k(x))\rangle\\
& = \langle D_{\rho'}(\alpha)(x),k(x)\rangle'\\
& = \rho'(k(x))(\alpha)\\
& = \rho(\varphi(k(x)))(\alpha\circ\psi)\\
& = \rho(\widehat{k}(\varphi_0(x)))(\widehat{\alpha})\\
& = \langle D_\rho(\widehat{\alpha})(\varphi_0(x)),\widehat{k}(\varphi(x))\rangle.
\end{align*}
From the nondegeneracy of $\langle\cdot,\cdot\rangle$ on $\mathrm{im}\,\varphi$, we conclude that $\ell\circ\varphi_0 = D_\rho(\widehat{\alpha})\circ\varphi_0$. In particular, we have
\begin{align*}
\varphi\circ D_{\rho'}(\langle f,g\rangle') & = D_\rho(\langle\widehat{f},\widehat{g}\rangle)\circ\varphi_0\\
& = [\![\widehat{f},\widehat{g}]\!]\circ\varphi_0+[\![\widehat{g},\widehat{f}]\!]\circ\varphi_0\\
& = [\![f,g]\!]'+[\![g,f]\!]'.
\end{align*}
This shows that $(\rho',\langle\cdot,\cdot\rangle',[\![\cdot,\cdot]\!]')$ is a Courant algebroid structure on $E'$.

The uniqueness of the Courant algebroid structure is a direct consequence of Proposition~\ref{prop:char_with_phi_rel_sec}.
\end{proof}

\section{port-Hamiltonian systems}

As a motivation, we consider the simple port-Hamiltonian system
\begin{equation}\label{eq:simple_pH}
\begin{aligned}
\tfrac{\mathrm{d}}{\mathrm{d}t}x & = J\nabla H(x) + Bu\\
y & = B^\top\nabla H(x)
\end{aligned}
\end{equation}
on the Euclidean space $\mathbb{R}^n$ with $J = -J^\top\in\mathbb{R}^{n\times n}$, $H\in\mathcal{C}^\infty(\mathbb{R}^n)$ and $B\in\mathbb{R}^{n\times m}$. This system is an open system which contains the Poisson-Hamilton system
\begin{align}\label{eq:PoisH}
\tfrac{\mathrm{d}}{\mathrm{d}t}x & = J\nabla H(x)
\end{align}
given by the ``no-interaction'' input $u\equiv 0$. Consider the Poisson-Hamilton system
\begin{align}\label{eq:interaction_HS}
\tfrac{\mathrm{d}}{\mathrm{d}t}\begin{pmatrix}
x\\z
\end{pmatrix} = \begin{bmatrix}
J & B\\
-B^\top & 0
\end{bmatrix}\nabla \mathcal{H}(t,x,z).
\end{align}
When we consider the family of interaction Hamiltonians $H_u(t,x,z) := H(x)+u(t)^\top z$, then we see that the behaviour, i.e. all solutions $(x,z)$, of~\eqref{eq:interaction_HS} with any of the interaction Hamiltonians $H_u$ (which is evidently contained in the behaviour of~\eqref{eq:interaction_HS} with any Hamiltonian) can be projected by $(x,z)\mapsto (x,-\tfrac{\mathrm{d}}{\mathrm{d}t}z)$ to the behaviour of~\eqref{eq:simple_pH}; and the latter contains the behaviour of the Poisson-Hamilton system~\eqref{eq:PoisH}. Therefore, the port-Hamiltonian system lies ``in between'' the closed systems~\eqref{eq:PoisH} and~\eqref{eq:interaction_HS}. This observation on the behaviours of the systems should carry over to the underlying geometric structures, in particular when we consider port-Hamiltonian systems on (nontrivial) manifolds whose port-space is given by a vector bundle. We therefore propose the following definition of an intrinsic geometric structure of port-Hamiltonian systems.

\begin{Definition}\label{def:intrinsic}
Let $\pi:E\to M$ be a vector bundle and consider a vector bundle isomorphism $\Phi$ between $\mathcal{T}E\oplus_E\mathcal{T}^*E$ and $p_1: E\oplus_M\mathcal{T}M\oplus_M\mathcal{T}^*M\oplus _ME\oplus_M E^*\to E$ over the identity on $E$. Equip $\mathcal{T}E\oplus_E\mathcal{T}^*E$ and $\mathcal{T}M\oplus_M\mathcal{T}^*M$ with their standard Courant algebroid structures. An intrinsic geometric structure of a port-Hamiltonian system on $M$ with port-structure $E$ given by $\Phi$ is each Courant-algebroid structure on $\mathcal{T}M\oplus_M\mathcal{T}^*M\oplus _ME\oplus_M E^*$ so that each arrow of
\begin{center}
\begin{tikzcd}
\mathcal{T}M\oplus_M\mathcal{T}^*M\arrow[d,"\mathrm{id}\oplus 0_{E\oplus E^*}"]\\
\mathcal{T}M\oplus_M\mathcal{T}^*M\oplus _ME\oplus_M E^*\arrow[d,"0_E\oplus\mathrm{id}"]\\
E\oplus_M\mathcal{T}M\oplus_M\mathcal{T}^*M\oplus _ME\oplus_M E^*\arrow[d,"\Phi"]\\
\mathcal{T}E\oplus_E\mathcal{T}^*E
\end{tikzcd}
\end{center}
is a classical Courant algebroid morphism, where $E\oplus_M\mathcal{T}M\oplus_M\mathcal{T}^*M\oplus _ME\oplus_M E^*$ has any Courant algebroid structure.
\end{Definition}

It may be shown that there exists at most one intrinsic structure given by a splitting $\Phi$.

\begin{Proposition}
With the notations of Definition~\ref{def:intrinsic}, the intrinsic geometric structure given by $\Phi$ is unique if it exists.
\end{Proposition}
\begin{proof}
The bases of the morphism in Definition~\eqref{def:intrinsic} are either zero sections or diffeomorphisms. Therefore, the prerequisites of Proposition~\ref{prop:pullback_algebroid} are met. Hence, there exists unique Courant-algebroid structures on $\mathcal{T}M\oplus_M\mathcal{T}^*M$, $\mathcal{T}M\oplus_M\mathcal{T}^*M\oplus _ME\oplus_M E^*$ and $E\oplus_M\mathcal{T}M\oplus_M\mathcal{T}^*M\oplus _ME\oplus_M E^*$ so that the arrows of the diagram in Definition~\ref{def:intrinsic} are Courant algebroid morphisms.
\end{proof}

\begin{Remark}
In Definition~\ref{def:intrinsic}, we have used an isomorphism $\mathcal{T}E\oplus_E\mathcal{T}^*E\overset{\Phi}{\cong} E\oplus_M\mathcal{T}M\oplus_M\mathcal{T}^*M\oplus _ME\oplus_M E^*$ of vector bundles. Up to an(other) isomorphism, such an isomorphism may be obtained from the (noncanonical!) splitting of $\mathcal{T}E$ into horizontal and vertical bundle, and the associated splitting of $\mathcal{T}^*E$ into the annihilator bundles, see e.g.~\cite{KolaMichSlov93}, so that there always exists a splitting $\Phi$ as proposed. In future works, we shall study which isomorphism define an intrinsic geometric structure of port-Hamiltonian systems.
\end{Remark}

\section{Conclusion}

We have studied Courant-algebroid structures on the bundle $\mathcal{T}M\oplus\mathcal{T}^*M\oplus E\oplus E^*$, which are related (by morphisms) to the standard Courant-algebroid structures on the Pontryagin bundles $\mathcal{T}M\oplus\mathcal{T}^*M$ and $\mathcal{T}E\oplus\mathcal{T}^*E$. Due to the definition of a morphism, the images (and pre-images) of isotropic structures are isotropic, relating the (almost Dirac structures of) port-Hamiltonian and its encompassing Hamiltonian systems. Thus, the terminology of Definition~\ref{def:intrinsic} appears not unjustified.

\bibliographystyle{alpha}
\bibliography{Literatur} 
\end{document}